\RequirePackage[l2tabu, orthodox]{nag}
\RequirePackage{snapshot}

\documentclass[10pt,onecolumn]{extarticle}

\makeatletter
\if@twocolumn
  \usepackage[dvips,letterpaper,top=0.5in, bottom=0.5in, left=0.75in, right=0.5in,includefoot,heightrounded]{geometry}
\else
  \usepackage[dvips,letterpaper,margin=1in,includefoot,heightrounded]{geometry}
\fi

\usepackage{srcltx}

\usepackage[russian,portuges,english]{babel}

\iflanguage{portuges}
    {\newcommand{\keywordname}{Palavras-chaves}}
    {\newcommand{\keywordname}{Keywords}}

\usepackage{amsmath,amssymb,amsfonts, bm}
\usepackage{mathdots}

\usepackage{abstract}

\usepackage{graphicx}
\usepackage[usenames,dvipsnames,svgnames,x11names]{xcolor}
\usepackage{subfigure}
\usepackage{psfrag}

\usepackage{booktabs}

\usepackage{setspace}
\usepackage{flushend}

\usepackage{cite}

\usepackage{hyperref}
\usepackage[normalem]{ulem}

\usepackage{enumerate}

\usepackage{multirow}

\usepackage[short,12hr]{datetime}
\usepackage[squaren]{SIunits}

       \usepackage{fouriernc}
\newtheorem{proposition}{Proposition}

\newtheorem{lemma}{Lemma}

\makeatletter

\newcommand{\printtitle}{%
\makeatletter
\if@twocolumn

\twocolumn[%
  \maketitle
  \begin{onecolabstract}
    \myabstract
  \end{onecolabstract}
  \begin{center}
    \small
    \textbf{\keywordname}
    \\\medskip
    \mykeywords
  \end{center}
  \bigskip
]
\saythanks
\else
  \maketitle
  \begin{onecolabstract}
    \myabstract
  \end{onecolabstract}
  \begin{center}
    \small
    \textbf{\keywordname}
    \\\medskip
    \mykeywords
  \end{center}
  \bigskip
  \onehalfspacing
\fi
\makeatother
}

\author{%
D.~F.~G.~Coelho%
\thanks{D.~F. G. Coelho is an independent researcher, Canada.
E-mail: \url{diegofgcoelho@gmail.com}%
}
\and
R.~J.~Cintra%
\thanks{%
R. J. Cintra is with the
Signal Processing Group,
Departamento de Estat\'{\i}stica,
Universidade Federal de Pernambuco, Brazil.
E-mail: \url{rjdsc@de.ufpe.br}%
}
\and
A.~Madanayake%
\thanks{%
A.~Madanayake is
with the
Department of Electrical and Computer Engineering,
Florida International University, FL.
E-mail: \url{amadanay@fiu.edu}
}
\and
S.~Perera%
\thanks{%
S.~M.~Perera is
with the
Department of Mathematics,
Embry-Riddle Aeronautical University, FL.
E-mail: \url{pereras2@erau.edu}
}
}

\title{%
Low-complexity Scaling Methods for \mbox{DCT-II} Approximations}

\newcommand{\myabstract}{%
This paper introduces
a collection of scaling methods
for generating $2N$-point \mbox{DCT-II} approximations
based on $N$-point low-complexity transformations.
Such scaling
is based on
the
Hou
recursive
matrix factorization
of the
exact
$2N$-point \mbox{DCT-II} matrix.
Encompassing the widely employed Jridi-Alfalou-Meher scaling method,
the proposed techniques
are shown to produce \mbox{DCT-II} approximations that
outperform the transforms resulting from
the
JAM scaling method
according to
total error energy and mean squared error.
Orthogonality conditions are derived
and an extensive
error analysis based on statistical simulation
demonstrates
the good performance of the introduced scaling methods.
A hardware implementation is also provided
demonstrating the competitiveness
of the proposed methods
when compared to the JAM scaling method.
}

\newcommand{\mykeywords}{%
DCT Approximation,
Fast algorithms,
Image compression
}

\date{}

\begin{document}

\printtitle

{A}{fter}
its inception in 1974~\cite{Ahmed1974},
the discrete cosine transform (DCT) has attracted
a significant amount of attention
leading to a large number of variations and algorithms
for its computation~\cite{Britanak2007}.
In particular,
the DCT of type II (\mbox{DCT-II}) is
widely employed for image
and video compression~\cite{Britanak2007}
codecs such as \mbox{H.262}/MPEG-2~\cite{IternationaTelecommunicationUnionITU2000},
\mbox{H.263}~\cite{IternationaTelecommunicationUnionITU2005},
\mbox{H.264}/AVC~\cite{Bhaskaran1997},
\mbox{H.265}/HEVC~\cite{Pourazad2012},
and the new \mbox{H.266}/VVC standard~\cite{Ochoa-Dominguez2019}.
The polynomial arithmetic technique---the use of the divide-and-conquer technique to reduce the degree of the
polynomial~\cite{Puschel2003, Puschel2008, Veoronenko2009, Steidl1991, Sandryhaila2011}---and
the matrix factorization technique---direct factorization
into the product of
sparse
matrices~\cite{Wang1984, Wang1985, Yip1980, Plonka2005, Perera2018a, Perera2018, Perera2015a, Perera2015, Perera2020}---are
the two main techniques that can be used to factor DCT matrices
and produce efficient hardware implementations.
Traditional video and image codecs
achieve
data decorrelation
by means of
blockwise image analysis
of $8\times 8$ subimages.
As a result,
considerable efforts have been put into
finding methods for
the 8-point \mbox{DCT-II}
computation~\cite{Haweel2001, Lengwehasatit2004,Bouguezel2008a,Cintra2011, Cintra2012, Potluri2012, Wahid2006, Liang2001}.
However,
the increasing demand for energy-efficient
hardware implementations moved the scientific community
towards the development of
low-complexity integer-based approximate transforms,
i.e.,
transforms
whose fast algorithms simply require trivial multiplications
(e.g.,~$\{0, \pm 1, \pm 1/2, \pm 2\}$)
and still provide good mathematical
properties~\cite{Haweel2001, Lengwehasatit2004,Bouguezel2008a,Cintra2011}.
In particular,
deriving good
low-complexity
approximations
has been an actively pursued
goal~\cite{Bouguezel2008, Bouguezel2009, Bouguezel2010, Bouguezel2011, Bouguezel2013, Bayer2011, Potluri2014}.
Such approximate approaches have led to
robust image and video compression systems
capable of significantly reducing
the computational cost of
the transform stage.
Although the 8-point \mbox{DCT-II} is ubiquitously employed,
modern codecs,
such as the
high efficiency video coding (HEVC),
require the \mbox{DCT-II} computation for larger blocklengths,
including
16- and~32-point transforms~\cite{Pourazad2012}.
Such demand
adds an extra difficulty for deriving good \mbox{DCT-II} approximations.

In its most general form,
designing
an $N$-point
approximate transform
requires
the solution of a non-linear multivariate
integer optimization problem
constrained to
the mathematical properties
of the desired
$N \times N$ matrix~\cite{Oliveira2019}.
Because it is often
an analytically intractable problem,
exhaustive search and brute force calculations
are
commonly
adopted as solution methods~\cite{Ehrgott2005, Oliveira2019}.
For small blocklengths, such as $N = 8$,
the exhaustive search
approach
is effective
and led to the proposition
of several low-complexity
approximations for the \mbox{DCT-II}~\cite{Coelho2018, Silveira2016}.
On the other hand,
as~$N$ increases,
the number of variables increases quadratically,
which
leads
to large computation times;
thus becoming an impractical approach.

To address such issue in designing~$N$-point \mbox{DCT-II} approximations,
matrix scaling has been considered as a feasible approach.
In the context of this paper,
the term
\emph{scaling}
refers to any procedure
capable of
deriving a larger transformation matrix
in terms of smaller transformations of the same kind.
In~\cite{Jridi2015},
Jridi, Alfalou, and Meher (JAM) proposed
a scaling method
for
deriving
$2N$-point \mbox{DCT-II} approximations
based on $N$-point \mbox{DCT-II} approximations.
Although JAM scaling method is a workable solution,
it consists of an \emph{ad-hoc} method
and
a clear analytical justification for its operation
is still lacking in the literature.
Moreover,
the JAM scaling
is an inexact procedure
in the sense that the scaling operation itself
introduces
errors
that
are
not due to the original $N$-point approximation.
Nevertheless,
due to the fact that
(i)~the \mbox{DCT-II} matrix is
highly symmetrical~\cite[p.~61, 70]{Britanak2007},
(ii)~the transform-based compression methods
are generally robust to matrix perturbations,
and
(iii)~image compression partially relies
on subjective aspects of the human visual system,
the approximations derived from the JAM scaling method
tended to provide practical transforms for image compression.
However,
we notice
that
there is a mathematical gap
in the understanding
of \mbox{DCT-II} scaling methods,
where more comprehensive
matrix analyses
are required.

The main goal of this paper is two-fold.
First we aim
at
providing a solid mathematical justification
to the
\mbox{DCT-II} approximation scaling proposed in ~\cite{Jridi2015}.
Second,
a collection of
scaling methods
capable of extending current approaches
for DCT approximations
which render low-complexity hardware implementations
is sought.
Direct matrix
factorizations~\cite{Wang1984, Wang1985, Yip1980, Plonka2005, Perera2018a, Perera2018, Perera2015a, Perera2015, Perera2020}
and
a
relation between
the \mbox{DCT-II} and the discrete sine (DST) transform of type IV (DST-IV)
are
employed to derive the sought methods.
As contributions, we also provide
16- and 32-point
\mbox{DCT-II} approximations
with
the associated
performance analysis compared to the scaled approximations obtained from the JAM method.

The paper is organized as follows.
Section~\ref{sec:background}
reviews
the DCT and DST definitions
and
relationships among them;
the JAM scaling method for \mbox{DCT-II} approximations
is also revisited.
Section~\ref{sec:scaling-dct-ii}
shows
the matrix derivation for obtaining
the
proposed recursive
algorithm to
compute the $2N$-point \mbox{DCT-II}.
Section~\ref{sec:scaling-methods}
introduces a family of scaling methods for
\mbox{DCT-II} approximations
in which
the JAM scaling method
is identified as a particular case.
Sufficient
mathematical conditions
for
the orthogonality of the scaled approximations
are also examined.
Section~\ref{sec:scaling-analysis}
presents error and coding performance analyses
as well as
an arithmetic complexity
evaluation of the proposed
family of scaling methods.
A suit of 16-point \mbox{DCT-II} approximations
resulting from the application of the proposed methods
to well-known 8-point \mbox{DCT-II} approximations
is also introduced and assessed.
Section~\ref{sec:conclusions}
brings final comments and future work~directions.

\section{Mathematical Background}
\label{sec:background}

\subsection{Discrete Cosine and Sine Transforms}

There are four main variants of DCT and DST,
which ranges from
types I to IV based on
Dirichlet and Neumann boundary conditions~\cite[p.~29--36]{Britanak2007}.
The entries
of the~$N$-point \mbox{DCT-II}, DCT-IV,
and DST-IV
transformation matrices
are, respectively, given by~\cite{Britanak2007}
\begin{align*}
\left[\mathbf{C}^\text{II}_N\right]_{k,n}
=
\sqrt{\frac{2}{N}}\beta_k
\cos\left( \frac{k(2n+1)\pi}{2N}\right)
,
\end{align*}
\begin{align*}
\left[\mathbf{C}^\text{IV}_N\right]_{k,n}
=
\sqrt{\frac{2}{N}}
\cos\left( \frac{(2k+1)(2n+1)\pi}{4N}\right)
,
\end{align*}
and
\begin{align*}
\left[\mathbf{S}^\text{IV}_N\right]_{k,n}
=
\sqrt{\frac{2}{N}}
\sin\left( \frac{(2k+1)(2n+1)\pi}{4N}\right),
\end{align*}
where~$k,n=0,1,\ldots,N-1$,
$\beta_0 = 1/\sqrt{2}$, and~$\beta_k = 1$,
for~\mbox{$k \neq 0$}.

\subsection{JAM Scaling for \mbox{DCT-II} Approximations}

Let
$\hat{\mathbf{C}}^\text{II}_{N}$
be an approximation for the
$N$-point
\mbox{DCT-II} matrix.
The JAM scaling method~\cite{Jridi2015}
generates
a $2N$-point \mbox{DCT-II} matrix approximation
$\hat{\mathbf{C}}^\text{II}_{2N}$
by using two instantiations of
a given $N$-point
\mbox{DCT-II} approximation matrix according to
\begin{align}
\label{equation-jam-scaling}
\hat{\mathbf{C}}^\text{II}_{2N} & =
\mathbf{P}_{2N}
\cdot
\begin{bmatrix}
\hat{\mathbf{C}}^\text{II}_{N}  &  \\
& \hat{\mathbf{C}}^\text{II}_{N} \end{bmatrix}
\cdot
\begin{bmatrix} \mathbf{I}_N  & \bar{\mathbf{I}}_N \\
\bar{\mathbf{I}}_N & -\mathbf{I}_N \end{bmatrix},
\end{align}
where
$\mathbf{I}_N$
and
$\bar{\mathbf{I}}_N$
are
the identity and counter-identity matrices of size~$N$,
respectively.
The matrix $\mathbf{P}_{2N}$
is
the
permutation
whose entries
are unitary
only at positions $(p_n, n)$,
where
\begin{align*}
p_n
=
\begin{cases}
2n, & n = 0, 1, \ldots, N - 1,
\\
2n (\operatorname{mod} 2N) + 1, & n = N, N+1, \ldots, 2N-1
.
\end{cases}
\end{align*}
The permutation~$\mathbf{P}_{2N}$
is
sometimes referred to
as the perfect shuffle~\cite[p.~66]{bracewell1985hartley},\cite{diaconis1983mathematics}
or the transpose of the even-odd permutation matrix~\cite{Blahut2010}.
In other word,
the JAM scaling method
is
the
mapping
described by:
\begin{align*}
\begin{split}
f_\text{JAM}:
\mathbb{R}^{N^2}
\longrightarrow
&
\mathbb{R}^{(2N)^2}
\\
\hat{\mathbf{C}}^\text{II}_{N}
\longmapsto
&
\hat{\mathbf{C}}^\text{II}_{2N}.
\end{split}
\end{align*}

The scaling expression \eqref{equation-jam-scaling}
stems
from
the following exact expression
based on an odd-even decomposition~\cite{Jridi2015}:
\begin{align}
\label{equation-original-exact-scaling}
\mathbf{C}_{2N}^\text{II}
& =
\frac{\sqrt{2}}{2}
\cdot
\mathbf{P}_{2N}
\cdot
\begin{bmatrix}
\mathbf{C}^\text{II}_N  &  \\
& \mathbf{J}_N \cdot \mathbf{S}^\text{IV}_N
\end{bmatrix}
\cdot
\begin{bmatrix}
\mathbf{I}_N  & \bar{\mathbf{I}}_N \\
\bar{\mathbf{I}}_N & -\mathbf{I}_N
\end{bmatrix}
,
\end{align}
where
~$
\mathbf{J}_N = \operatorname{diag} \left( \{(-1)^n\}_{n=0}^{N-1} \right)
$.
In order
to introduce~\eqref{equation-jam-scaling},
the authors
in~\cite{Jridi2015},
have
(i)~replaced
the
the lower-right block
$\mathbf{J}_N \cdot \mathbf{S}^\text{IV}_N$
matrix in~\eqref{equation-original-exact-scaling}
with
the $N$-point \mbox{DCT-II} matrix;
and
then
(ii)~substituted the exact \mbox{DCT-II} matrices with
approximate \mbox{DCT-II} matrices.

\subsection{Relationships between \mbox{DCT-II} and DCT-IV}
The
\mbox{DCT-II} computation
can be performed
by means of the
Chen algorithm~\cite{Wen-HsiengChen1977}.
Such algorithm
is based on a decomposition
that expresses~$\mathbf{C}^\text{II}_{2N}$
in terms of~$\mathbf{C}^\text{II}_N$
and~$\mathbf{C}^\text{IV}_N$~\cite{Britanak2013, Plonka2005, Perera2015, Perera2018, Puschel2008, Suehiro1986, Hsu2008}.
The matrix form of the Chen algorithm
is given
by~\cite[p.~96]{Britanak2007}%
\footnote{%
Equation 4.50 in~\cite[p.~96]{Britanak2007}
misses
the~${\sqrt{2}}/{2}$ factor.}
\begin{align*}
\mathbf{C}^\text{II}_{2N}
=
\frac{\sqrt{2}}{2}
\cdot
\mathbf{R}_{2N}
\cdot
\begin{bmatrix}
\mathbf{R}_N \cdot \mathbf{C}^\text{II}_N & \\
& \mathbf{R}_N \cdot \mathbf{C}^\text{IV}_N \cdot \bar{\mathbf{I}}_N \end{bmatrix}
\cdot
\begin{bmatrix}
\mathbf{I}_N  & \bar{\mathbf{I}}_N \\
\bar{\mathbf{I}}_N & -\mathbf{I}_N
\end{bmatrix}
,
\end{align*}
where~$\mathbf{R}_N$ is
the bit-reversal ordering
permutation matrix~\cite{Britanak2007, Blahut2010}.
The above expression can be simplified
by noticing that,
for any
integer~$N$,
the perfect shuffle
and
the bit-reversal permutations
satisfy the following expression:
\begin{align*}
\mathbf{P}_{2N}
=
\mathbf{R}_{2N} \cdot
\begin{bmatrix}
\mathbf{R}_N  &  \\
& \mathbf{R}_N
\end{bmatrix}
.
\end{align*}
Therefore,
we have
the matrix form
given below~\cite{Wang1984}:
\begin{align}
\label{equation-chen-factorization-simplified}
\mathbf{C}^\text{II}_{2N}
=
\frac{\sqrt{2}}{2}
\cdot
\mathbf{P}_{2N}
\cdot
\begin{bmatrix}
\mathbf{C}^\text{II}_N & \\
& \mathbf{C}^\text{IV}_N \cdot \bar{\mathbf{I}}_N
\end{bmatrix}
\cdot
\begin{bmatrix}
\mathbf{I}_N  & \bar{\mathbf{I}}_N \\
\bar{\mathbf{I}}_N & -\mathbf{I}_N
\end{bmatrix}
.
\end{align}

In~\cite{Chan1990} and~\cite{Kok1997},
it has been demonstrated
that
the matrices
$\mathbf{C}^\text{II}_N$ and~$\mathbf{C}^\text{IV}_N$ are
related according
to~\cite[p.~77]{Britanak2007}, \cite{Plonka2005, Puschel2008, Suehiro1986, Hsu2008}
\begin{align}
\label{eq:rel_A}
\mathbf{C}^\text{IV}_N =  \mathbf{A}_N \cdot \mathbf{C}^\text{II}_N \cdot \mathbf{D}_N,
\end{align}
where
\begin{align*}
\mathbf{D}_N
&
=
\operatorname{diag}
\left(
\left\{
2\cos\left( \frac{(2n+1)\pi}{4N} \right)
\right\}_{n=0}^{N-1}
\right)
,
\end{align*}
\begin{align}
\label{eq:A_N}
\mathbf{A}_N =
\mathbf{J}_N
\cdot
\operatorname{tril}
\left(
\mathbf{u}_N
\cdot
\begin{bmatrix}
\frac{\sqrt{2}}{2} & \mathbf{u}_{N-1}^\top
\end{bmatrix}
\right)
\cdot
\mathbf{J}_N
,
\end{align}
$\mathbf{u}_N$ is the $N$-point
column vector of ones
and
$\operatorname{tril}(\cdot)$
returns the lower triangular part
of its matrix argument
setting all other entries to zero~\cite{Shores2007}.
Hereafter,
let
$\mathbf{U}_N
=
\mathbf{u}_N
\cdot
\begin{bmatrix}
\frac{\sqrt{2}}{2} & \mathbf{u}_{N-1}^\top
\end{bmatrix}$.

\subsection{Recursive Computation of the \mbox{DCT-II}}
\label{sec:scaling-dct-ii}

The proposed
scaling method
relies on
finding
a suitable
relation between the~$2N$-point \mbox{DCT-II} and the $N$-point \mbox{DCT-II}.
We seek
a scaling expression
capable of
(i)~taking advantage of the \mbox{DCT-II}
regular
structures
and
(ii)~encompassing
the
JAM scaling method as a particular case.
Thus,
we aim at preserving the butterfly stage
characterized
in the rightmost matrices
in~\eqref{equation-original-exact-scaling}
and~\eqref{equation-chen-factorization-simplified}.
The following proposition due to Hou~\cite{HsiehHou1987},
to which we offer an alternative proof,
establishes the sought
relationship between the~$2N$-point \mbox{DCT-II} and the $N$-point \mbox{DCT-II}.

\begin{proposition}
\label{theo:exact-scaling}
The $2N$-point \mbox{DCT-II} matrix can be factored in the form
\begin{align}
\label{eq:scalable_general}
\mathbf{C}^\text{II}_{2N}
=
\frac{\sqrt{2}}{{2}}
\cdot
\mathbf{P}_{2N}
&
\cdot
\begin{bmatrix}
\mathbf{I}_N  &  \\
 & \mathbf{B}_N
\end{bmatrix}
\cdot
\begin{bmatrix}
\mathbf{C}^\text{II}_N & \\
& \mathbf{C}^\text{II}_N
\end{bmatrix}
\nonumber
\\
&
\cdot
\begin{bmatrix}
\mathbf{I}_N  &  \\
 & \mathbf{G}_N
\end{bmatrix}
\cdot
\begin{bmatrix}
\mathbf{I}_N  & \bar{\mathbf{I}}_N \\
\bar{\mathbf{I}}_N  & -\mathbf{I}_N
\end{bmatrix},
\end{align}
where
\begin{align*}
\mathbf{B}_N
=
-\bar{\mathbf{I}}_N \cdot
\operatorname{tril}
\left(\mathbf{U}_N
\right) \cdot \mathbf{J}_N
\end{align*}
and
\begin{align*}
\mathbf{G}_N
=
\operatorname{diag}
\left(
\left\{
2(-1)^{n}\cos\left( \frac{(2n+1)\pi}{4N}\right)
\right\}_{n = 0}^{N-1}
\right).
\end{align*}
\end{proposition}

\begin{proof}
In~\cite{Wang1984},
Wang has demonstrated the following
relationship
between
the DST-IV
and
the DCT-IV~\cite{Britanak2007}:
\begin{align}
\label{eq:s-iv-relation-to-c-iv}
\mathbf{S}^\text{IV}_N
=
\bar{\mathbf{I}}_N
\cdot
\mathbf{C}^\text{IV}_N
\cdot
\mathbf{J}_N
.
\end{align}
The DST-IV and DCT-IV
are related according
to~\cite{Plonka2005}:
\begin{align}
\label{eq:DCTIV-counter-identity}
\mathbf{C}^\text{IV}_N
\cdot
\bar{\mathbf{I}}_N
=
\mathbf{J}_N\cdot
\mathbf{S}^\text{IV}_N
.
\end{align}
Combining the above expression with~\eqref{eq:s-iv-relation-to-c-iv},
we
obtain:
\begin{align*}
\mathbf{C}^\text{IV}_N\cdot
\bar{\mathbf{I}}_N
=
\mathbf{J}_N\cdot
\bar{\mathbf{I}}_N
\cdot
\mathbf{C}^\text{IV}_N
\cdot
\mathbf{J}_N.
\end{align*}
Replacing $\mathbf{C}^\text{IV}_N$ by~\eqref{eq:rel_A}
yields the expression below:
\begin{align*}
\mathbf{C}^\text{IV}_N
\cdot
\bar{\mathbf{I}}_N
=
\mathbf{J}_N
\cdot
\bar{\mathbf{I}}_N
\cdot
\mathbf{A}_N
\cdot
\mathbf{C}^\text{II}_N
\cdot
\mathbf{D}_N
\cdot
\mathbf{J}_N.
\end{align*}
Setting~$\mathbf{B}_N =\mathbf{J}_N \cdot \bar{\mathbf{I}}_N \cdot \mathbf{A}_N$ and~$\mathbf{G}_N
=\mathbf{D}_N
\cdot
\mathbf{J}_N
$,
we get
\begin{align}
\mathbf{C}^\text{IV}_N
\cdot
\bar{\mathbf{I}}_N
=
\mathbf{B}_N
\cdot
\mathbf{C}^\text{II}_N
\cdot
\mathbf{G}_N
.
\label{ne}
\end{align}
Here we applied
the result
$
\mathbf{J}_N
\cdot
\bar{\mathbf{I}}_N
=-\bar{\mathbf{I}}_N \cdot \mathbf{J}_N
$
and
then
we
combined
the entries of the diagonal matrices
$\mathbf{D}_N$ and $\mathbf{J}_N$
to obtain $\mathbf{G}_N$.
Also, we used~\eqref{eq:A_N}
to obtain that
$\mathbf{B}_N =\mathbf{J}_N \cdot \bar{\mathbf{I}}_N \cdot \mathbf{A}_N = -\bar{\mathbf{I}}_N \cdot \operatorname{tril} \left(\mathbf{U}_N \right) \cdot \mathbf{J}_N$.
Replacing
the
term
$\mathbf{C}^\text{IV}_N\cdot\bar{\mathbf{I}}_N$
in~\eqref{equation-chen-factorization-simplified}
by~\eqref{ne}
yields~\eqref{eq:scalable_general}.
\end{proof}

The structure of the matrix ~$\mathbf{B}_N$ is explicitly given by:
\begin{align*}
\mathbf{B}_N
&
=
-\bar{\mathbf{I}}_N \cdot
\operatorname{tril}
\left(\mathbf{U}_N
\right) \cdot \mathbf{J}_N
=
\left[
\begin{smallmatrix}
-\frac{\sqrt{2}}{2} & 1 & -1 & \ldots & 1 & -1 & \phantom{-}1\\
-\frac{\sqrt{2}}{2} & 1 & -1 & \ldots & 1 & -1 &  \\
-\frac{\sqrt{2}}{2} & 1 & -1 & \ldots & 1 &  &  \\
\vdots & \vdots & \vdots & \iddots &  & \\
\vdots & \vdots & \iddots & &  & \\
-\frac{\sqrt{2}}{2} & 1 &  &  &  &  & \\
-\frac{\sqrt{2}}{2} &  &   &  &  &  &
\end{smallmatrix}
\right]
.
\end{align*}

\section{Approximate Scaling}
\label{sec:scaling-methods}

\subsection{The Family of Scaling Methods}

In general,
a \mbox{DCT-II} approximation
$\hat{\mathbf{C}}^\text{II}_{N}$
can be represented according to
the polar decomposition~\cite[p.~348]{seber2008matrix}
consisting of two parts:
(i)~a low-complexity matrix~$\mathbf{T}_N$
and
(ii)~a diagonal matrix~$\bm{\Sigma}_N$
that provides
orthogonalization or
quasi-orthogonalization~\cite{Coelho2018, seber2008matrix}.
Such matrices are related
according
to:
\begin{align*}
\hat{\mathbf{C}}^\text{II}_{N}
=
\bm{\Sigma}_N
\cdot
\mathbf{T}_N
,
\end{align*}
where
\begin{align*}
\bm{\Sigma}_N
=
\sqrt{
\operatorname{diag}
\left\{
\left(
\mathbf{T}_N
\cdot
\mathbf{T}_N^\top
\right)^{-1}
\right\}
}
\end{align*}
and
$\sqrt{\cdot}$
is the matrix square root~\cite{seber2008matrix}.
Here the operator $\operatorname{diag}(\cdot)$
returns a diagonal matrix
with the elements of the diagonal
of
its matrix argument~\cite{seber2008matrix, Coelho2018}.

Considering
the Proposition~\ref{theo:exact-scaling},
we introduce the
following
mapping
relating
an $N$-point
low-complexity matrix~$\mathbf{T}_N$
to
its $2N$-point scaled form~$\mathbf{T}_{2N}$
as shown below:
\begin{align}
\label{eq:scalable_general_app}
\begin{split}
f_{\left(\hat{\mathbf{B}}_N, \hat{\mathbf{G}}_N\right)}:
\mathbb{R}^{N^2}
\longrightarrow
&
\mathbb{R}^{(2N)^2}
\\
\mathbf{T}_N
\longmapsto
&
\mathbf{T}_{2N}
=
\mathbf{P}_{2N}
\cdot
\begin{bmatrix}
\mathbf{I}_N & \\
&   \hat{\mathbf{B}}_N
\end{bmatrix}
\cdot
\begin{bmatrix}
\mathbf{T}_N & \\
&   \mathbf{T}_N
\end{bmatrix}
\\
&
\phantom{\mathbf{T}_{2N}=}
\cdot
\begin{bmatrix}
\mathbf{I}_N & \\
&   \hat{\mathbf{G}}_N
\end{bmatrix}
\cdot
\begin{bmatrix}
\mathbf{I}_N  & \bar{\mathbf{I}}_N \\
\bar{\mathbf{I}}_N & -\mathbf{I}_N
\end{bmatrix}.
\end{split}
\end{align}
Matrices~$\hat{\mathbf{B}}_N$
and~$\hat{\mathbf{G}}_N$
are parameter matrices.
Approximations
for~$\mathbf{B}_N$
and~$\mathbf{G}_N$
appear
as
natural candidates for the parameter matrices
$\hat{\mathbf{B}}_N$
and~$\hat{\mathbf{G}}_N$,
respectively.
The actual $2N$-point approximate \mbox{DCT-II}
is obtained after orthogonalization~\cite{Tablada2015, Coelho2018}
and
is given by:
\begin{align}
\label{eq:orthogonalization}
\hat{\mathbf{C}}^\text{II}_{2N}
=
\bm{\Sigma}_{2N}
\cdot
\mathbf{T}_{2N}
,
\end{align}
where
$
\bm{\Sigma}_{2N}
=
\sqrt{
\operatorname{diag}
\left\{
\left(
\mathbf{T}_{2N}
\cdot
\mathbf{T}_{2N}^\top
\right)^{-1}
\right\}
}$.
Although
\eqref{eq:scalable_general_app}
stems from
\eqref{eq:scalable_general}
(Proposition~\ref{theo:exact-scaling}),
it does not need to inherit
the scalar~$\sqrt{2}/2$.
This is because
the orthogonalization
in \eqref{eq:orthogonalization}
is invariant
to the presence of constant scalars~\cite{Tablada2015}.
Ultimately,
depending on
the choice
of the parameter
matrices~$\hat{\mathbf{B}}_N$ and~$\hat{\mathbf{G}}_N$,
we obtain different methods for scaling \mbox{DCT-II} approximations.

Selected choices of
parameter matrices~$\hat{\mathbf{B}}_N$ and~$\hat{\mathbf{G}}_N$
are
shown
in Table~\ref{tab:scaling_families},
where~$\mathbf{Z}_N = \operatorname{diag}(\begin{bmatrix}1/2&\mathbf{u}_{N-1}\end{bmatrix})$.
The list is not complete as~\emph{any choice}
of~$\hat{\mathbf{B}}_N$ and/or~$\hat{\mathbf{G}}_N$
that are regarded to be
`close enough' to~$\mathbf{B}_N$ and~$\mathbf{G}_N$, respectively,
generates a particular scaling method.
For instance,
the JAM scaling method proposed in~\cite{Jridi2015}
is
the particular case when
$\hat{\mathbf{B}}_N = \hat{\mathbf{G}}_N = \mathbf{I}_N$,
furnishing the relationship below:
\begin{align*}
\hat{\mathbf{C}}^\text{II}_{2N}
=
f_\text{JAM}
(\hat{\mathbf{C}}^\text{II}_{N})
=
f_{(\mathbf{I}_N,\mathbf{I}_N)}
(\hat{\mathbf{C}}^\text{II}_{N})
.
\end{align*}
Notice that the coefficients of the matrices were intentionally chosen to be small magnitude integers
which are in the set~$\{0, \pm 1/2, \pm 1, \pm 2\}$.
The faithfulness of the~$2N$-point DCT approximation will come as a result of
how the entries of the parameter matrices~$\hat{\mathbf{B}}_N$ and~$\hat{\mathbf{G}}_N$
are chosen.
Other choices of parameter matrices could be obtained by bit-expanding
the original matrices~${\mathbf{B}}_N$ and~${\mathbf{G}}_N$
or performing multicriteria optimization over the coefficients
of~$\hat{\mathbf{B}}_N$ and~$\hat{\mathbf{G}}_N$~\cite{Coelho2018, Tablada2015} and
taking into account the specifics of the application in hand.

\begin{table}
\centering
\small
\caption{Scaling Approximations and Its Performance}
\label{tab:scaling_families}
\begin{tabular}{c@{\quad} c@{\quad} c@{\quad} c@{\quad} c@{\quad} c@{\quad} c@{\quad} c@{\quad} c@{\quad}}\toprule
\multirow{2}{*}{Case} &\multirow{2}{*}{$\hat{\mathbf{B}}_N$} &  \multirow{2}{*}{$\hat{\mathbf{G}}_N$} & \multicolumn{3}{c}{$\left| \hat{\mathbf{C}}^\text{II}_{2N}-\mathbf{C}^\text{II}_{2N} \right|_\mathsf{F}$} & & \multirow{2}{*}{\small{Orth.?}} \\\cmidrule{4-6}
& & & $N=8$ & $N=16$ & $N=32$ & & \\\midrule
JAM & $\mathbf{I}_N$ & $\mathbf{I}_N$ & $3.994$ & $5.653$ & $7.997$ & & Yes\\\midrule

I & $\bar{\mathbf{I}}_N$ & $\mathbf{I}_N$ & $3.826$ & $5.533$ & $7.912$ & & Yes\\\midrule

II & $-\bar{\mathbf{I}}_N\cdot\mathbf{J}_N$ & $\mathbf{I}_N$ & $4.001$ & $5.657$ & $8.000$ & & Yes\\\midrule

III & $-\bar{\mathbf{I}}_N\cdot\mathbf{Z}_N\cdot\mathbf{J}_N $ & $\mathbf{I}_N$ & $4.001$ & $5.657$ & $8.000$ & & Yes\\\midrule

IV & $\mathbf{I}_N$ & $\mathbf{J}_N$ & $3.826$ & $5.533$ & $7.912$ & & Yes \\\midrule

V & $\bar{\mathbf{I}}_N$ & $\mathbf{J}_N$ & $4.006$ & $5.661$ & $8.003$ & & Yes\\\midrule

VI & $-\bar{\mathbf{I}}_N\cdot\mathbf{J}_N$ & $\mathbf{J}_N$ & $1.954$ & $3.033$ & $4.515$ & & Yes\\\midrule

VII & $-\bar{\mathbf{I}}_N\cdot\mathbf{Z}_N\cdot\mathbf{J}_N $ & $\mathbf{J}_N$ & $1.954$ & $3.033$ & $4.515$ & & Yes \\\midrule

\end{tabular}
\end{table}

\subsection{Orthogonality Condition}

The design
of approximate transforms
often
require
the transformation matrices to be orthogonal~\cite{Britanak2007, Tablada2015, Coelho2018}.
In fact,
in contexts
such as
noise reduction~\cite{Gupta2012},
watermarking methods~\cite{An2009},
and harmonic detection~\cite{Limin2007,Zheng2010, Britanak2007},
the invertibility of the DCT is not only desired, but required.
This is because the DCT is
used to translate the signal to its transform domain,
where processing is performed.
The resulting
signal in the transform domain is then translated back into the time domain,
rendering the final desired output.
Thus, we show in Proposition \ref{proposition-orthogonality}, the sufficient conditions but not the necessary conditions for orthogonally of the proposed DCT-II approximation.

It can be shown~\cite{Tablada2015}
that
if
$\mathbf{T}_{2N}\cdot\mathbf{T}_{2N}^\top$ is a diagonal matrix,
then
the
\mbox{DCT-II} approximation~$\hat{\mathbf{C}}^\text{II}_{2N}$
obtained according to~\eqref{eq:orthogonalization}
is
orthogonal.
Using~\eqref{eq:scalable_general_app},
we can write
\begin{equation}
\begin{split}
\mathbf{T}_{2N}\cdot \mathbf{T}_{2N}^\top
=
&
2
\cdot
\mathbf{P}_{2N}
\\
&
\cdot
\begin{bmatrix}
\mathbf{T}_N\cdot\mathbf{T}_N^\top & \\
&
\hat{\mathbf{B}}_N
\cdot
\mathbf{T}_N
\cdot
\hat{\mathbf{G}}_N
\cdot
\hat{\mathbf{G}}_N^\top
\cdot
\mathbf{T}_N^\top
\cdot
\hat{\mathbf{B}}_N^\top
\end{bmatrix}
\\
&
\cdot
\mathbf{P}_{2N}
.
\end{split}
\label{eq:T2N-times-T2N-transpose}
\end{equation}
Thus,
for
$\mathbf{T}_{2N}\cdot \mathbf{T}_{2N}^\top$
to be a diagonal matrix,
we have
to ensure that
both matrices
$\mathbf{T}_N\cdot\mathbf{T}_N^\top$
and
$
\hat{\mathbf{B}}_N
\cdot
\mathbf{T}_N
\cdot
\hat{\mathbf{G}}_N
\cdot
\hat{\mathbf{G}}_N^\top
\cdot
\mathbf{T}_N^\top
\cdot
\hat{\mathbf{B}}_N^\top$
are diagonal matrices.
In order to investigate the conditions for orthogonality,
we use the following result from~\cite[p.~151]{seber2008matrix}.
\begin{lemma}
If~$\mathbf{P}$ is a permutation matrix
and~$\mathbf{D}$ is a diagonal matrix,
then~$\mathbf{P} \cdot \mathbf{D} \cdot \mathbf{P}^\top$
is a diagonal matrix.
\label{theo:PDP}
\end{lemma}

Therefore,
sufficient conditions
for orthogonality
are furnished by the proposition below.

\begin{proposition}
\label{proposition-orthogonality}
If the following conditions are satisfied:
\begin{enumerate}[(i)]
\item
$\mathbf{T}_N\cdot\mathbf{T}_N^\top$ is a diagonal matrix;
\item
$\hat{\mathbf{G}}_N\cdot\hat{\mathbf{G}}_N^\top
=
a\cdot\mathbf{I}_N$,
$a\in\mathbb{R}$;
\item
$\hat{\mathbf{B}}_N$ is
a generalized permutation matrix;
\end{enumerate}
then
the scaling method in~\eqref{eq:scalable_general_app}
generates an orthogonal \mbox{DCT-II} approximation.

\begin{proof}
We need to ensure that the sub-matrices
from the block-diagonal matrix in~\eqref{eq:T2N-times-T2N-transpose}
are diagonal matrices.
Therefore,
the Condition~(i) is
clearly a necessary condition.
Now let us examine the lower-right sub-matrix
$
\hat{\mathbf{B}}_N
\cdot
\mathbf{T}_N
\cdot
\hat{\mathbf{G}}_N
\cdot
\hat{\mathbf{G}}_N^\top
\cdot
\mathbf{T}_N^\top
\cdot
\hat{\mathbf{B}}_N^\top
$.
Using~\eqref{eq:T2N-times-T2N-transpose} and
Condition~(ii), we obtain:
\begin{align*}
\hat{\mathbf{B}}_N
\cdot
\mathbf{T}_N
\cdot
\hat{\mathbf{G}}_N
\cdot
\hat{\mathbf{G}}_N^\top
\cdot
\mathbf{T}_N^\top
\cdot
\hat{\mathbf{B}}_N^\top
=
a
\cdot
\hat{\mathbf{B}}_N
\cdot
\mathbf{T}_N
\cdot
\mathbf{T}_N^\top
\cdot
\hat{\mathbf{B}}_N^\top
.
\end{align*}
Condition~(iii) ensures that
$
\hat{\mathbf{B}}_N
=
\mathbf{D}
\cdot
\mathbf{P}
$,
where~$\mathbf{D}$ is a diagonal matrix
and~$\mathbf{P}$ is a permutation matrix.
Thus,
we have that:
$
\hat{\mathbf{B}}_N
\cdot
\mathbf{T}_N
\cdot
\mathbf{T}_N^\top
\cdot
\hat{\mathbf{B}}_N^\top
=
\mathbf{D}
\cdot
\mathbf{P}
\cdot
\mathbf{T}_N
\cdot
\mathbf{T}_N^\top
\cdot
\mathbf{P}^\top
\cdot
\mathbf{D}
$.
By applying Lemma~\ref{theo:PDP},
it follows that
$
\mathbf{P}
\cdot
\mathbf{T}_N
\cdot
\mathbf{T}_N^\top
\cdot
\mathbf{P}^\top
$
is a diagonal matrix.
We have then that
$
\hat{\mathbf{B}}_N
\cdot
\mathbf{T}_N
\cdot
\mathbf{T}_N^\top
\cdot
\hat{\mathbf{B}}_N^\top
=
\mathbf{D}
\cdot
\left(
\mathbf{P}
\cdot
\mathbf{T}_N
\cdot
\mathbf{T}_N^\top
\cdot
\mathbf{P}^\top
\right)
\cdot
\mathbf{D}$
is also a diagonal matrix.
\end{proof}
\end{proposition}

Notice that the conditions
required by
Proposition~\ref{proposition-orthogonality}
are not too restrictive.
In fact,
because
the exact matrix~$\mathbf{G}_N$ is
by definition
a diagonal matrix,
the condition on~$\hat{\mathbf{G}}_N$
(Condition (ii))
can be met by approximating the elements
of
$\mathbf{G}_N$
to
a suitable
value
$\pm\sqrt{a}$
(e.g., $a=1$).
The methods listed in Table~\ref{tab:scaling_families}
are capable of
generating orthogonal approximations,
because
the selected choices for
$\hat{\mathbf{B}}_N$ and~$\hat{\mathbf{G}}_N$
are under
the conditions
prescribed in Proposition~\ref{proposition-orthogonality}.

\section{Error, Performance, and Complexity Analysis}
\label{sec:scaling-analysis}

\subsection{Error Analysis and Statistical Modeling}

In order to assess the proposed scaling methods,
we performed
an error analysis
based on
the Frobenius norm of the
difference~\mbox{$\hat{\mathbf{C}}^\text{II}_{2N}-\mathbf{C}^\text{II}_{2N}$}.
To properly
isolate
the behavior
of the scaling method,
the required
$N$-point matrices were ensured to be identical
and equal to the exact DCT matrix,
i.e.,
$\hat{\mathbf{C}}^\text{II}_N = \mathbf{C}^\text{II}_N$.
Table~\ref{tab:scaling_families}
shows
the computed errors
for $N \in \{ 8, 16, 32 \}$.
In all cases,
Methods I, IV, VI, and~VII
generated smaller errors,
outperforming
the JAM scaling method.

Now let us analyze
the errors when
actual approximations
$\hat{\mathbf{C}}^\text{II}_{N}$
are considered,
i.e.
$
\hat{\mathbf{C}}^\text{II}_N
\not=
\mathbf{C}^\text{II}_N
$.
The performance
of the proposed scaling methods
can be quantified
by means
of
the error of~$\hat{\mathbf{C}}^\text{II}_{2N}$,
relative to~$\mathbf{C}^\text{II}_{N}$,
as a function of the error of~$\hat{\mathbf{C}}^\text{II}_{N}$,
relative to~$\mathbf{C}^\text{II}_{N}$.
Because of the wide popularity and importance of
the 8-point \mbox{DCT-II},
we fixed $N=8$
as the most relevant case for analysis.
Thus,
the values of~$\|\hat{\mathbf{C}}^\text{II}_{8}-\mathbf{C}^\text{II}_{8}\|_{\mathsf{F}}$
were computed for
the 8-point \mbox{DCT-II} approximations
discussed below.
The recent book~\cite{Ochoa-Dominguez2019} by Rao,
co-inventor of the DCT,
identifies state-of-art transforms
such as
the series of approximations
by Bouguezel-Ahmad-Swamy (BAS)~\cite[p.~160]{Ochoa-Dominguez2019},
the rounded DCT (RDCT)~\cite[p.~162]{Ochoa-Dominguez2019},
and
the
modified RDCT (MRDCT)~\cite[p.~162]{Ochoa-Dominguez2019}.
Due to its flexibility,
we selected the BAS parametric approximation described
in~\cite{Bouguezel2011}
for
parameter values $a=0, 1/2, 1$,
referred to as
$\text{BAS}_1$,
$\text{BAS}_2$,
and
$\text{BAS}_3$.
We also included the BAS approximation~\cite{Bouguezel2013}
labeled here as $\text{BAS}_4$.
In~\cite{Lombardi2018},
the RDCT
and
the MRDCT
were
identified as optimal approximations
in terms of output image quality and computing time, respectively,
according to a hardware implementation using approximate adder cells
for
image compression.
In addition to the above-mentioned approximations,
we included
in our comparisons
the very recently introduced
angle-based approximate DCT in~\cite{Oliveira2019}
(here termed ABDCT),
the traditional signed DCT (SDCT)~\cite{Haweel2001},
the Lengwehasatit-Ortega DCT
approximation (LODCT)~\cite{Lengwehasatit2004},
and
the low-complexity approximation detailed in~\cite{Potluri2014},
which is an improved version of the MRDCT,
here denoted as IMRDCT.

The error for the corresponding scaled approximations
$\|\hat{\mathbf{C}}^\text{II}_{16}-\mathbf{C}^\text{II}_{16}\|_{\mathsf{F}}$
against
$\|\hat{\mathbf{C}}^\text{II}_{8}-\mathbf{C}^\text{II}_{8}\|_{\mathsf{F}}$
for each of the methods in Table~\ref{tab:scaling_families}
are shown in Figure~\ref{fig:C2NError}.
The errors follow a linear trend
that can be quantified
according
to
a linear regression
using
least-square estimation~\cite{Kay1993, Kay1998}
for
the linear model below:
\begin{align*}
g(\|\hat{\mathbf{C}}^\text{II}_{16}-\mathbf{C}^\text{II}_{16}\|_{\mathsf{F}})
=
m\cdot \|\hat{\mathbf{C}}^\text{II}_{8}-\mathbf{C}^\text{II}_{8}\|_{\mathsf{F}}+b
,
\end{align*}
where
$m$ and $b$ are
the slope and intercept to be estimated,
respectively.
Table~\ref{tab:linear-regression} shows
the estimates of~$m$ and~$b$, $\hat{m}$ and~$\hat{b}$,
along with
the~$\chi^2$ goodness of fit statistic
and
the residual mean squared error (RMSE)
for the model.
At the significance level of~$0.001$,
the critical value
was
approximately~$20.1$
for all scenarios.
The
regression models presented
$\chi^2$ test statistic values
smaller  than~$\approx 7.9\cdot 10^{-2}$ %
and
RMSE values
smaller  than~$\approx 1.0\cdot 10^{-1}$.
Such values
are much smaller than the critical value
for the test at~$0.001$ significance level;
thus
we have $p$-values very close to the unit~\cite{casellaberger},
preventing us from rejecting the models.
\begin{figure*}
\centering

\psfrag{norm(CNh-CN)}[][][0.8]{$\|\hat{\mathbf{C}}^\text{II}_{8}-\mathbf{C}^\text{II}_{8}\|_{\mathsf{F}}$}
\psfrag{norm(C2Nh-C2N)}[][][0.8]{$\|\hat{\mathbf{C}}^\text{II}_{16}-\mathbf{C}^\text{II}_{16}\|_{\mathsf{F}}$}

\subfigure[Method~JAM]{\includegraphics[scale=0.90]{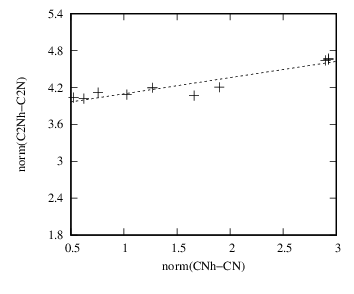}}
\subfigure[Method~I]{\includegraphics[scale=0.90]{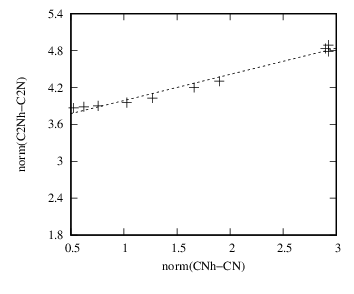}}
\subfigure[Method~II]{\includegraphics[scale=0.90]{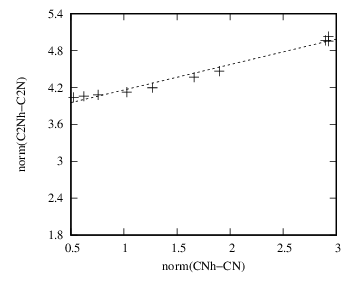}}
\subfigure[Method~III]{\includegraphics[scale=0.90]{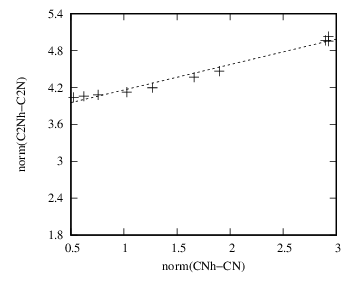}}
\subfigure[Method~IV]{\includegraphics[scale=0.90]{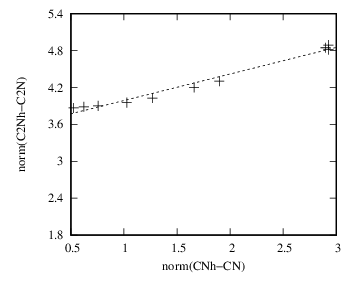}}
\subfigure[Method~V]{\includegraphics[scale=0.90]{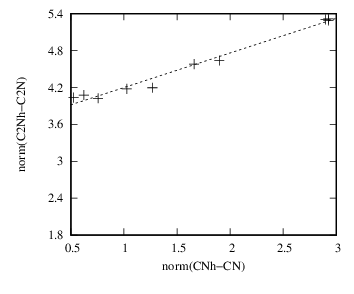}}
\subfigure[Method~VI]{\includegraphics[scale=0.90]{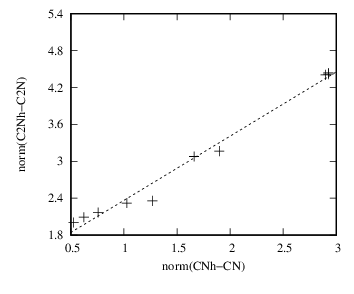}}
\subfigure[Method~VII]{\includegraphics[scale=0.90]{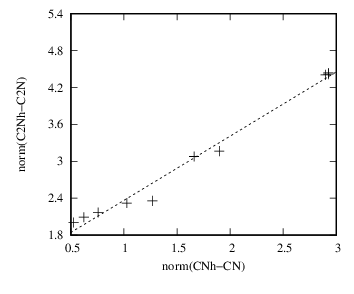}}
\caption{$\|\hat{\mathbf{C}}^\text{II}_{16}-\mathbf{C}^\text{II}_{16}\|_{\mathsf{F}}$ as
a function of $\|\hat{\mathbf{C}}^\text{II}_{8}-\mathbf{C}^\text{II}_{8}\|_{\mathsf{F}}$
for the methods outlined in Table~\ref{tab:scaling_families}.}
\label{fig:C2NError}
\end{figure*}

\begin{table}
\centering
\small
\caption{Linear regression analysis statistics
using least-squares estimator}
\begin{tabular}{c@{\quad}c@{\quad}c@{\quad}c@{\quad}c@{\quad}}\toprule
Method & $\hat{m}$ & $\hat{b}$ & $\chi^2$ & RMSE \\
\midrule
JAM & $0.264$ & $3.833$ & $7.979\cdot 10^{-2}$ & $9.987\cdot 10^{-2}$ \\\midrule
I & $0.426$ & $3.561$ & $3.599\cdot 10^{-2}$ & $6.708 \cdot 10^{-2}$\\\midrule
II & $0.413$ & $3.746$ & $3.177\cdot 10^{-2}$& $6.243 \cdot 10^{-2}$\\\midrule
III & $0.413$ & $3.746$ & $3.177\cdot 10^{-2}$& $6.243 \cdot 10^{-2}$\\\midrule
IV & $0.431$ & $3.555$ & $3.636\cdot 10^{-2}$ & $6.742 \cdot 10^{-2}$\\\midrule
V & $0.562$ & $3.636$ & $5.220\cdot 10^{-2}$ & $8.077 \cdot 10^{-2}$\\\midrule
VI & $1.045$ & $1.319$ & $1.531\cdot 10^{-1}$& $1.383 \cdot 10^{-1}$\\\midrule
VII & $1.045$ & $1.319$ & $1.531\cdot 10^{-1}$& $1.383 \cdot 10^{-1}$\\\bottomrule
\end{tabular}
\label{tab:linear-regression}
\end{table}

The quantity~$\hat{m}$
determines the average
influence of
\mbox{$\|\hat{\mathbf{C}}^\text{II}_{8}-\mathbf{C}^\text{II}_{8}\|_{\mathsf{F}}$}
over~\mbox{$\|\hat{\mathbf{C}}^\text{II}_{16}-\mathbf{C}^\text{II}_{16}\|_{\mathsf{F}}$};
whereas $\hat{b}$
captures
the minimum error due to the scaling method.
No matter how good the approximation~$\hat{\mathbf{C}}^\text{II}_{8}$ is,
the resulting approximation~$\hat{\mathbf{C}}^\text{II}_{16}$
has
an inherent error due to the
approximations~$\hat{\mathbf{B}}_8$ and~$\hat{\mathbf{G}}_{8}$.
This floor error
is equal to the intercept~$\hat{b}$.
The JAM method results
in the lowest~$\hat{m}$
but also in the highest~$\hat{b}$;
whereas the proposed
Method~VI and~VII %
presents
the highest~$\hat{m}$ and lowest~$\hat{b}$.

We can compare two fitted models
by determining the
crossing points of the curves
representing the linear models
and slope of error regression curve~\mbox{$\|\hat{\mathbf{C}}^\text{II}_{16}-\mathbf{C}^\text{II}_{16}\|_{\mathsf{F}}$}
as a function of~\mbox{$\|\hat{\mathbf{C}}^\text{II}_{8}-\mathbf{C}^\text{II}_{8}\|_{\mathsf{F}}$}
for each model.
Thus,
for instance,
the Method~VII provides better approximations
when compared to
the JAM method
if~$\|\hat{\mathbf{C}}^\text{II}_{8}-\mathbf{C}^\text{II}_{8}\| < 3.219$.
Table~\ref{tab:break-point}
shows
the maximum value
of~$\|\hat{\mathbf{C}}^\text{II}_{8}-\mathbf{C}^\text{II}_{8}\|$
for
which
the JAM scaling method is outperformed
by each of the proposed methods.
The 8-point \mbox{DCT-II} approximations
in the literature
present
Frobenius
errors
in the range $[1.72, 2.68]$.
Therefore,
Methods~VI and~VII outperform
JAM method
regardless
of
the considered
8-point \mbox{DCT-II} approximation.
Similar analyses can be applied to larger approximations.

\begin{table}

\centering
\small
\caption{Approximate
$\|\hat{\mathbf{C}}^\text{II}_{8}-\mathbf{C}^\text{II}_{8}\|$ maximum values
for which the proposed methods outperform
the JAM method}

\begin{tabular}{c@{\quad}c@{\quad}c@{\quad}c@{\quad}c@{\quad}c@{\quad}c@{\quad}c@{\quad}}\toprule
Method & I & II & III & IV & V & VI & VII \\\midrule
$\|\hat{\mathbf{C}}^\text{II}_{8}-\mathbf{C}^\text{II}_{8}\|$ & $1.679$ & $0.584$ &  $1.664$ &  $1.664$ &  $0.661$ &  $3.219$ &  $3.219$ \\\bottomrule
\end{tabular}

\label{tab:break-point}
\end{table}

\subsection{Performance Measurements}

We assessed the performance
of the obtained 16-point
DCT approximations
according
to
the following
figures of merit:
the
mean-squared
error~$\operatorname{MSE}(\cdot)$~\cite{Britanak2007, Cham1989, Malvar1992},
total error energy~$\epsilon(\cdot)$~\cite{Britanak2007, Cintra2011},
deviation from
orthogonality~$d(\cdot)$~\cite{Britanak2007, Cintra2012},
unified coding gain~$C_g(\cdot)$~\cite{Britanak2007, Cham1989},
and
transform efficiency~$\eta(\cdot)$~\cite{Britanak2007, Liang2001}.
The total error energy quantifies
the error between matrices
in a euclidean distance way~\cite{Britanak2007, Cham1989, Malvar1992}, while the
mean square error (MSE)
of a given matrix approximation takes into account its proximity to the original
transform and its effect on the autocorrelation matrix of the class of signals in consideration~\cite{Cintra2011, Tablada2015}.
The unified transform coding gain~\cite{Britanak2007, Cham1989} and transform efficiency~\cite{Britanak2007, Liang2001}
provide measures
to quantify the compression capabilities
of a given approximation~\cite{Tablada2015}.
Tables~\ref{tab:metrics_SDCT}--\ref{tab:metrics_T_14_potluri}
show
the obtained results.

The JAM scaling method
is outperformed
by Methods~VI, and~VII
in terms of
total error energy
when considering
the
$\text{BAS}_1$, $\text{BAS}_2$, $\text{BAS}_3$, $\text{BAS}_4$,
SDCT,
LO,
RDCT,
MRDCT,
ABDCT,
and
IMRDCT
as shown in
Tables~\ref{tab:metrics_BAS6}--\ref{tab:metrics_T_14_potluri}.
Methods~I, II and~III
have consistently offered
the same coding performance
when compared with the JAM scaling method
under
all considered scenarios.
\begin{table}
\centering
\small
\caption{Metrics for scaling methods using $\text{BAS}_1$ transform}
\label{tab:metrics_BAS6}
\begin{tabular}{c@{\quad}c@{\quad}c@{\quad}c@{\quad}c@{\quad}c@{\quad}c@{\quad}c@{\quad}}
\\\toprule
Method & $d(\cdot)$ & $\epsilon(\cdot)$ & $\operatorname{MSE}(\cdot)$ & $C_g(\cdot)$ & $\eta(\cdot)$ & $A(\cdot)$ & $S(\cdot)$\\\midrule
JAM & $0.00$ & $14.62$ & $0.14$ & $8.16$ & $70.98$ & 48 & 0\\\midrule
I & $0.00$ & $15.04$ & $0.34$ & $8.16$ & $70.98$ & 48 & 0\\\midrule
II & $0.00$ & $15.79$ & $0.35$ & $8.16$ & $70.98$ & 48 & 0\\\midrule
III & $0.00$ & $15.79$ & $0.35$ & $8.16$ & $70.98$ & 48 & 0\\\midrule
IV & $0.00$ & $15.13$ & $0.36$ & $7.16$ & $57.36$ & 48 & 0\\\midrule
V & $0.00$ & $16.62$ & $0.42$ & $7.16$ & $57.36$ & 48 & 0\\\midrule
VI & $0.00$ & $13.88$ & $0.40$ & $7.16$ & $57.36$ & 48 & 0\\\midrule
VII & $0.00$ & $13.88$ & $0.40$ & $7.16$ & $57.36$ & 48 & 0\\\bottomrule
\end{tabular}
\end{table}

\begin{table}
\centering
\small
\caption{Metrics for scaling methods using $\text{BAS}_2$ transform}
\label{tab:metrics_BAS7}
\begin{tabular}{c@{\quad}c@{\quad}c@{\quad}c@{\quad}c@{\quad}c@{\quad}c@{\quad}c@{\quad}}
\\\toprule
Method & $d(\cdot)$ & $\epsilon(\cdot)$ & $\operatorname{MSE}(\cdot)$ & $C_g(\cdot)$ & $\eta(\cdot)$ & $A(\cdot)$ & $S(\cdot)$\\\midrule
JAM & $0.00$ & $14.58$ & $0.14$ & $8.37$ & $71.83$ & 52 & 4\\\midrule
I & $0.00$ & $15.19$ & $0.35$ & $8.37$ & $71.83$ & 52 & 4\\\midrule
II & $0.00$ & $15.61$ & $0.36$ & $8.37$ & $71.83$ & 52 & 4\\\midrule
III & $0.00$ & $15.61$ & $0.36$ & $8.37$ & $71.83$ & 52 & 4\\\midrule
IV & $0.00$ & $15.23$ & $0.37$ & $7.48$ & $58.83$ & 52 & 4\\\midrule
V & $0.00$ & $16.67$ & $0.44$ & $7.48$ & $58.83$ & 52 & 4\\\midrule
VI & $0.00$ & $13.84$ & $0.42$ & $7.48$ & $58.83$ & 52 & 4\\\midrule
VII & $0.00$ & $13.84$ & $0.42$ & $7.48$ & $58.83$ & 52 & 4\\\bottomrule
\end{tabular}
\end{table}

\begin{table}
\centering
\small
\caption{Metrics for scaling methods using $\text{BAS}_3$ transform}
\label{tab:metrics_BAS5}
\begin{tabular}{c@{\quad}c@{\quad}c@{\quad}c@{\quad}c@{\quad}c@{\quad}c@{\quad}c@{\quad}}
\\\toprule
Method & $d(\cdot)$ & $\epsilon(\cdot)$ & $\operatorname{MSE}(\cdot)$ & $C_g(\cdot)$ & $\eta(\cdot)$ & $A(\cdot)$ & $S(\cdot)$\\\midrule
JAM & $0.00$ & $14.67$ & $0.14$ & $8.16$ & $70.80$ & 52 & 0\\\midrule
I & $0.00$ & $15.36$ & $0.36$ & $8.16$ & $70.80$ & 52 & 0\\\midrule
II & $0.00$ & $15.57$ & $0.37$ & $8.16$ & $70.80$ & 52 & 0\\\midrule
III & $0.00$ & $15.57$ & $0.37$ & $8.16$ & $70.80$ & 52 & 0\\\midrule
IV & $0.00$ & $15.36$ & $0.37$ & $7.41$ & $59.95$ & 52 & 0\\\midrule
V & $0.00$ & $16.70$ & $0.44$ & $7.41$ & $59.95$ & 52 & 0\\\midrule
VI & $0.00$ & $13.94$ & $0.42$ & $7.41$ & $59.95$ & 52 & 0\\\midrule
VII & $0.00$ & $13.94$ & $0.42$ & $7.41$ & $59.95$ & 52 & 0\\\bottomrule
\end{tabular}
\end{table}

\begin{table}
\centering
\small
\caption{Metrics for scaling methods using $\text{BAS}_4$ transform}
\label{tab:metrics_BAS8}
\begin{tabular}{c@{\quad}c@{\quad}c@{\quad}c@{\quad}c@{\quad}c@{\quad}c@{\quad}c@{\quad}}
\\\toprule
Method & $d(\cdot)$ & $\epsilon(\cdot)$ & $\operatorname{MSE}(\cdot)$ & $C_g(\cdot)$ & $\eta(\cdot)$ & $A(\cdot)$ & $S(\cdot)$\\\midrule
JAM & $0.00$ & $13.18$ & $0.13$ & $8.19$ & $70.65$ & 64 & 0\\\midrule
I & $0.00$ & $12.65$ & $0.34$ & $8.19$ & $70.65$ & 64 & 0\\\midrule
II & $0.00$ & $13.18$ & $0.36$ & $8.19$ & $70.65$ & 64 & 0\\\midrule
III & $0.00$ & $13.18$ & $0.36$ & $8.19$ & $70.65$ & 64 & 0\\\midrule
IV & $0.00$ & $12.65$ & $0.34$ & $8.19$ & $70.65$ & 64 & 0\\\midrule
V & $0.00$ & $13.18$ & $0.13$ & $8.19$ & $70.65$ & 64 & 0\\\midrule
VI & $0.00$ & $7.40$ & $0.06$ & $8.19$ & $70.65$ & 64 & 0\\\midrule
VII & $0.00$ & $7.40$ & $0.06$ & $8.19$ & $70.65$ & 64 & 0\\\bottomrule
\end{tabular}
\end{table}

\begin{table}
\centering
\small
\caption{Metrics for scaling methods using RDCT}
\label{tab:metrics_RDCT}
\begin{tabular}{c@{\quad}c@{\quad}c@{\quad}c@{\quad}c@{\quad}c@{\quad}c@{\quad}c@{\quad}}
\\\toprule
Method & $d(\cdot)$ & $\epsilon(\cdot)$ & $\operatorname{MSE}(\cdot)$ & $C_g(\cdot)$ & $\eta(\cdot)$ & $A(\cdot)$ & $S(\cdot)$\\\midrule
JAM & $0.00$ & $12.93$ & $0.12$ & $8.43$ & $72.23$ & 60 & 0\\\midrule
I & $0.00$ & $12.25$ & $0.31$ & $8.43$ & $72.23$ & 60 & 0\\\midrule
II & $0.00$ & $12.82$ & $0.30$ & $8.43$ & $72.23$ & 60 & 0\\\midrule
III & $0.00$ & $12.82$ & $0.30$ & $8.43$ & $72.23$ & 60 & 0\\\midrule
IV & $0.00$ & $12.25$ & $0.34$ & $7.50$ & $59.87$ & 60 & 0\\\midrule
V & $0.00$ & $12.65$ & $0.14$ & $7.50$ & $59.87$ & 60 & 0\\\midrule
VI & $0.00$ & $6.80$ & $0.07$ & $7.50$ & $59.87$ & 60 & 0\\\midrule
VII & $0.00$ & $6.80$ & $0.07$ & $7.50$ & $59.87$ & 60 & 0\\\bottomrule
\end{tabular}
\end{table}

\begin{table}
\centering
\small
\caption{Metrics for scaling methods using MRDCT}
\label{tab:metrics_T_14_bayer}
\begin{tabular}{c@{\quad}c@{\quad}c@{\quad}c@{\quad}c@{\quad}c@{\quad}c@{\quad}c@{\quad}}
\\\toprule
Method & $d(\cdot)$ & $\epsilon(\cdot)$ & $\operatorname{MSE}(\cdot)$ & $C_g(\cdot)$ & $\eta(\cdot)$ & $A(\cdot)$ & $S(\cdot)$\\\midrule
JAM & $0.00$ & $12.77$ & $0.13$ & $7.58$ & $66.07$ & 44 & 0\\\midrule
I & $0.00$ & $13.19$ & $0.34$ & $7.58$ & $66.07$ & 44 & 0\\\midrule
II & $0.00$ & $13.72$ & $0.34$ & $7.58$ & $66.07$ & 44 & 0\\\midrule
III & $0.00$ & $13.72$ & $0.34$ & $7.58$ & $66.07$ & 44 & 0\\\midrule
IV & $0.00$ & $13.19$ & $0.36$ & $6.48$ & $52.20$ & 44 & 0\\\midrule
V & $0.00$ & $14.39$ & $0.25$ & $6.48$ & $52.20$ & 44 & 0\\\midrule
VI & $0.00$ & $9.67$ & $0.18$ & $6.48$ & $52.20$ & 44 & 0\\\midrule
VII & $0.00$ & $9.67$ & $0.18$ & $6.48$ & $52.20$ & 44 & 0\\\bottomrule
\end{tabular}
\end{table}

\begin{table}
\centering
\small
\caption{Metrics for scaling methods using the ABDCT}
\label{tab:metrics_RAZ}
\begin{tabular}{c@{\quad}c@{\quad}c@{\quad}c@{\quad}c@{\quad}c@{\quad}c@{\quad}c@{\quad}}
\\\toprule
Method & $d(\cdot)$ & $\epsilon(\cdot)$ & $\operatorname{MSE}(\cdot)$ & $C_g(\cdot)$ & $\eta(\cdot)$ & $A(\cdot)$ & $S(\cdot)$\\\midrule
JAM & $0.00$ & $12.63$ & $0.12$ & $8.88$ & $76.81$ & 64 & 12\\\midrule
I & $0.00$ & $12.21$ & $0.31$ & $8.88$ & $76.81$ & 64 & 12\\\midrule
II & $0.00$ & $12.75$ & $0.32$ & $8.88$ & $76.81$ & 64 & 12\\\midrule
III & $0.00$ & $12.75$ & $0.32$ & $8.88$ & $76.81$ & 64 & 12\\\midrule
IV & $0.00$ & $12.21$ & $0.34$ & $8.18$ & $63.79$ & 64 & 12\\\midrule
V & $0.00$ & $12.81$ & $0.14$ & $8.18$ & $63.79$ & 64 & 12\\\midrule
VI & $0.00$ & $6.56$ & $0.07$ & $8.18$ & $63.79$ & 64 & 12\\\midrule
VII & $0.00$ & $6.56$ & $0.07$ & $8.18$ & $63.79$ & 64 & 12\\\bottomrule
\end{tabular}
\end{table}

\begin{table}
\centering
\small
\caption{Metrics for scaling methods using SDCT}
\label{tab:metrics_SDCT}
\begin{tabular}{c@{\quad}c@{\quad}c@{\quad}c@{\quad}c@{\quad}c@{\quad}c@{\quad}c@{\quad}}
\\\toprule
Method & $d(\cdot)$ & $\epsilon(\cdot)$ & $\operatorname{MSE}(\cdot)$ & $C_g(\cdot)$ & $\eta(\cdot)$ & $A(\cdot)$ & $S(\cdot)$\\\midrule
JAM & $0.20$ & $12.83$ & $0.13$ & $6.27$ & $68.82$ & 64 & 0\\\midrule
I & $0.20$ & $12.42$ & $0.34$ & $6.27$ & $68.82$ & 64 & 0\\\midrule
II & $0.20$ & $12.96$ & $0.36$ & $6.27$ & $68.82$ & 64 & 0\\\midrule
III & $0.20$ & $12.96$ & $0.36$ & $6.27$ & $68.82$ & 64 & 0\\\midrule
IV & $0.20$ & $12.42$ & $0.38$ & $5.57$ & $58.11$ & 64 & 0\\\midrule
V & $0.20$ & $13.12$ & $0.16$ & $5.57$ & $58.11$ & 64 & 0\\\midrule
VI & $0.20$ & $7.29$ & $0.09$ & $5.57$ & $58.11$ & 64 & 0\\\midrule
VII & $0.20$ & $7.29$ & $0.09$ & $5.57$ & $58.11$ & 64 & 0\\\bottomrule
\end{tabular}
\end{table}

\begin{table}
\centering
\small
\caption{Metrics for scaling methods using the LODCT}
\label{tab:metrics_LO}
\begin{tabular}{c@{\quad}c@{\quad}c@{\quad}c@{\quad}c@{\quad}c@{\quad}c@{\quad}c@{\quad}}
\\\toprule
Method & $d(\cdot)$ & $\epsilon(\cdot)$ & $\operatorname{MSE}(\cdot)$ & $C_g(\cdot)$ & $\eta(\cdot)$ & $A(\cdot)$ & $S(\cdot)$\\\midrule
JAM & $0.00$ & $12.67$ & $0.12$ & $8.64$ & $73.11$ & 64 & 4\\\midrule
I & $0.00$ & $12.15$ & $0.30$ & $8.64$ & $73.11$ & 64 & 4\\\midrule
II & $0.00$ & $12.69$ & $0.31$ & $8.64$ & $73.11$ & 64 & 4\\\midrule
III & $0.00$ & $12.69$ & $0.31$ & $8.64$ & $73.11$ & 64 & 4\\\midrule
IV & $0.00$ & $12.15$ & $0.34$ & $7.83$ & $61.49$ & 64 & 4\\\midrule
V & $0.00$ & $12.68$ & $0.14$ & $7.83$ & $61.49$ & 64 & 4\\\midrule
VI & $0.00$ & $6.30$ & $0.07$ & $7.83$ & $61.49$ & 64 & 4\\\midrule
VII & $0.00$ & $6.30$ & $0.07$ & $7.83$ & $61.49$ & 64 & 4\\\bottomrule
\end{tabular}
\end{table}

\begin{table}
\centering
\small
\caption{Metrics for scaling methods using the IMRDCT
}
\label{tab:metrics_T_14_potluri}
\begin{tabular}{c@{\quad}c@{\quad}c@{\quad}c@{\quad}c@{\quad}c@{\quad}c@{\quad}c@{\quad}}
\\\toprule
Method & $d(\cdot)$ & $\epsilon(\cdot)$ & $\operatorname{MSE}(\cdot)$ & $C_g(\cdot)$ & $\eta(\cdot)$ & $A(\cdot)$ & $S(\cdot)$\\\midrule
JAM & $0.00$ & $13.21$ & $0.15$ & $7.58$ & $66.07$ & 44 & 0\\\midrule
I & $0.00$ & $13.51$ & $0.39$ & $7.58$ & $66.07$ & 44 & 0\\\midrule
II & $0.00$ & $14.03$ & $0.39$ & $7.58$ & $66.07$ & 44 & 0\\\midrule
III & $0.00$ & $14.03$ & $0.39$ & $7.58$ & $66.07$ & 44 & 0\\\midrule
IV & $0.00$ & $13.51$ & $0.37$ & $6.48$ & $52.20$ & 44 & 0\\\midrule
V & $0.00$ & $14.58$ & $0.26$ & $6.48$ & $52.20$ & 44 & 0\\\midrule
VI & $0.00$ & $9.94$ & $0.20$ & $6.48$ & $52.20$ & 44 & 0\\\midrule
VII & $0.00$ & $9.94$ & $0.20$ & $6.48$ & $52.20$ & 44 & 0\\\bottomrule
\end{tabular}
\end{table}

\subsection{Arithmetic Complexity}
The only matrix structures in~\eqref{eq:scalable_general_app}
that contribute to the arithmetic complexity are:
(i)~the two instantiations of~$\mathbf{T}_N$;
(ii)~the butterfly matrix of size~$2N$;
(iii)~the diagonal matrix~$\hat{\mathbf{G}}_N$;
and
(iv)~the matrix~$\hat{\mathbf{B}}_N$.
The matrices
$\hat{\mathbf{G}}_N$
and
$\hat{\mathbf{B}}_N$
requires no multiplication.
Thus,
if
$\mathbf{T}_N$
is selected to be a multiplierless
approximation,
then
the proposed scaling methods
are
ensured to have null multiplicative complexity.
Therefore
the arithmetic complexity
is fully characterized by
the
number of additions
and bit-shifting operations,
which
are
given by:
\begin{align*}
A(\mathbf{T}_{2N}) & = 2A(\mathbf{T}_N)+A(\hat{\mathbf{B}}_N)+A(\hat{\mathbf{G}}_N)+2N
\end{align*}
and
\begin{align*}
S(\mathbf{T}_{2N}) & = 2S(\mathbf{T}_N)+S(\hat{\mathbf{B}}_N)+S(\hat{\mathbf{G}}_N),
\end{align*}
where
functions $A(\cdot)$ and $S(\cdot)$
return the number of additions and bit-shifting operations
required by its arguments,
respectively~\cite{Blahut2010}.
Tables~\ref{tab:metrics_SDCT}--\ref{tab:metrics_T_14_potluri}
shows the arithmetic complexity
for the considered methods.
The proposed scaling methods does not incur
in higher arithmetic complexity when compared to
the JAM scaling method.

\section{Hardware Implementation}

The JAM method and the proposed methods
listed in Table~\ref{tab:scaling_families} were
implemented on
a field programmable gate array (FPGA).
The device used for the hardware implementation was
the Xilinx Artix-7
XC7A35T-1CPG236C.

Because of its high coding performance
(see~Table~\ref{tab:metrics_RAZ}),
we selected the ABDCT~\cite{Oliveira2019}
to be submitted to the discussed methods.
Each scaled transform~$\mathbf{T}_{2N}$
employed
two pipelined instances of the ABDCT core,
as described in Figure~\ref{fig:block_diagram} showing the internal architecture for the resulting transformation matrix~$\mathbf{T}_{2N}$.
For each
of the methods outlined in Table~\ref{tab:scaling_families},
the respective
matrices~$\hat{\mathbf{B}}_N$ and~$\hat{\mathbf{G}}_N$
are generalized permutations~\cite[p.~151]{seber2008matrix}.
Therefore,
the implementation of such matrices
solely requires combinational logic
leading to a reduced overall design latency.
Each sub-block implementing the ABDCT was implemented according
to the fast algorithm outlined in~\cite{Oliveira2019}, where
arithmetic operations were pipelined for achieving
higher maximum operating frequency of each block.
The architecture was implemented using input wordlength of~8~bits.

\begin{figure}

\centering
\includegraphics[scale=1.0]{general_scaling.eps}
\caption{Block diagram for the proposed scaling methods for DCT approximation. For the hardware implementation using the ABDCT, the~$\mathbf{T}_{N}$ sub-blocks implement the ABDCT.}
\label{fig:block_diagram}
\end{figure}

The designs were tested employing
the scheme depicted in Figure~\ref{fig:testbed}, together with a controller state-machine and connected to
a
universal asynchronous receiver-transmitter (UART) block.
The UART core interfaces
with the controller state machine
using the ARM Advanced Microcontroller Bus Architecture Advanced eXtensible Interface 4
(AMBA AXI-4) protocol.
\begin{figure}
\centering
\includegraphics[scale=0.7]{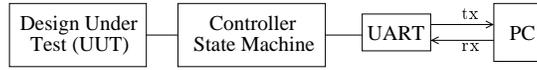}
\caption{Testbed architecture for testing the implemented designs.}
\label{fig:testbed}
\end{figure}
A personal computer (PC) communicates with the controller through the UART by
sending a set of 16~coefficients corresponding to the input submitted to the transform under evaluation.
The 16 coefficients are passed to the design
and processed,
then
the controller state machine
sends the resulting 16~coefficients
back
to the personal computer.
This operation is performed several times and then to the output
of processing the original coefficients with
a software model implemented in Python.

The metrics examined to evaluate the hardware implementations were:
number of occupied slices,
number of look-up tables~(LUT),
flip-flop~(FF) count,
critical path delay~($T_{\text{cpd}}$),
maximum operating frequency $F_{\text{max}} = T_{\text{cpd}}^{-1}$,
and
dynamic power~($D_p$) normalized by~$F_{\text{max}}$.

\begin{table*}
\centering
\caption{FPGA measures of the implemented architectures for scaling of ABDCT}
\label{tab:metrics_hardware}
\begin{tabular}{lrrrrrrrr}
\toprule
\multirow{2}{*}{Metric} & \multicolumn{8}{c}{Method} \\ \cmidrule{2-9}
& JAM & I & II & III & IV & V & VI & VII \\ \midrule
\# Slices			& 224 & 230 & 246 & 265 & 241 & 249 & 251 & 270\\
\# LUT				& 673 & 672 & 724 & 723 & 684 & 684 & 738 & 736\\
\# FF				& 1061& 1061& 1061& 1058& 1061& 1061& 1061& 1058\\
$T_{\text{cpd}}$~(\nano\second)	& 4.558 & 4.202 & 4.407 & 4.485 & 4.606 & 4.694 & 4.897 & 4.620 \\
$F_{\text{max}}$~(\mega\hertz) 	& 219.394 & 237.982 & 226.917 & 222.965 & 217.108 & 213.038 & 204.207 & 216.450 \\
$D_p$~(\micro\watt\per\mega\hertz) & 86.389 & 79.838 & 88.138 & 89.700 & 87.514 & 89.186 & 93.043 & 87.780\\ \bottomrule
\end{tabular}
\end{table*}

Table~\ref{tab:metrics_hardware} shows the hardware metrics for the implemented designs.
In terms of resources consumption, the JAM method requires the least amount of slices and LUTs.
Method VII, however, has the least amount of FFs, while using the largest amount of slices a high number of LUTs.

Method I is the one achieving the lowest critical path delay, followed by Method II, III, and then JAM, respectively.
Because of that, Method I, II, III, and JAM are the ones achieving the highest maximum operating frequency, respectively.
Method I is also the one achieving the most efficient implementation in terms of dynamic power, representing a reduction of 7.58\% when compared to the JAM implementation,
which is the second most efficient implementation.

\section{Conclusions}
\label{sec:conclusions}

We provided an alternative derivation to the
recursive algorithm proposed by Hou~\cite{HsiehHou1987}.
By judiciously approximating
specific matrix factors
of Hou recursive \mbox{DCT-II} factorization,
we
introduced
a framework
for approximate scaling
that
is
capable of
deriving
$2N$-point DCT approximations
based on
$N$-point DCT approximations.
The proposed
collection of scaling DCT approximations
is flexible
and
generates several methods,
encompassing
the
JAM scaling method
as
a particular case.
Conditions for orthogonality---a
common property
in the context of image/data compression---were identified.
An error analysis and statistical modeling
of the scaling methods were derived.
The proposed scaling methods are
inherently multiplierless,
i.e., they do not contribute to any multiplication.
An arithmetic complexity analysis was derived
and
expressions for the additive
and bit-shifting costs were furnished.
The new proposed scaling methods were
able to outperform the competing JAM method in terms of
Frobenius errors and coding gain,
paving the way for promising hardware implementations.

As a topic for future research,
the authors are aware that the work in~\cite{Perera2018b}
proposed a relationship between
the \mbox{DCT-II} and \mbox{DCT-IV} with matrix factors
that possess the highest
sparsity in the literature.
This relation can be used
in connection with~\eqref{eq:s-iv-relation-to-c-iv}
in place of~\eqref{eq:DCTIV-counter-identity},
leading to alternative derivations and
a different family of scaling methods.
Possible research fronts are also the investigation of
how the resulting DCT-II approximations behave after several uses of the recursion in~\eqref{eq:scalable_general_app}
and the use of matrix parameterizations based on the generalization in~\cite{WenjiaYuan2006}, which is derived from a generalization of different five algorithms in the literature~\cite{HsiehHou1987, Kok1997, Lee1984, PeiZongLee1994, Cvetkovic1992}.

\onecolumn

{\small
\singlespacing
\bibliographystyle{ieeetr}
\bibliography{ref}
}

\end{document}